\documentclass[journal,twocolumn]{IEEEtran}
\IEEEoverridecommandlockouts
\usepackage{cite}
\usepackage{amsmath,amssymb,amsfonts}
\usepackage{multirow}
\usepackage{graphicx}
\usepackage{subfig}
\usepackage{textcomp}
\usepackage{xcolor}
\usepackage[utf8]{inputenc}
\usepackage[english]{babel}
\usepackage{fancyhdr}
\usepackage[font=footnotesize,skip=0pt]{caption}
\usepackage{array}
\usepackage{tabularx}
\usepackage[ruled,vlined]{algorithm2e}
\usepackage{algpseudocode}
\usepackage{amsmath}
\usepackage{amsthm}
\usepackage{mathtools, nccmath} 

\newtheorem{lemma}{Lemma}

\newtheorem{corollary}{Corollary}

\usepackage[english]{babel}

\def\BibTeX{{\rm B\kern-.05em{\sc i\kern-.025em b}\kern-.08em
    T\kern-.1667em\lower.7ex\hbox{E}\kern-.125emX}}
\begin{document}

\title{User Pairing and Outage Analysis in Multi-Carrier  NOMA-THz Networks\\
}
\author{\IEEEauthorblockN{Sadeq Bani Melhem, \IEEEmembership{Student Member,~IEEE} and Hina Tabassum, \IEEEmembership{Senior Member,~IEEE}}

\vspace{-10mm}
\thanks{S.~B.Melhem and H.~Tabassum are with the  York University, Canada
(e-mail: sadeq@yorku.ca and hina@eecs.yorku.ca). This work is supported by the Discovery
Grant from the Natural Sciences and Engineering Research Council of Canada.}
}

\raggedbottom

\maketitle

\begin{abstract}
 This paper provides a comprehensive framework to analyze the performance of non-orthogonal multiple access (NOMA) in the downlink transmission of a  single-carrier and multi-carrier terahertz (THz) network. Specifically, we first develop a novel user pairing scheme for the THz-NOMA network which ensures the performance gains of NOMA over orthogonal multiple access (OMA) for each individual user in the NOMA pair and adapts according to the molecular absorption. Then, we characterize novel outage probability expressions considering a single-carrier and multi-carrier THz-NOMA network in the presence of various user pairing schemes, Nakagami-$m$ channel fading, and molecular absorption noise. We propose a moment-generating-function (MGF) based approach to analyze the outage probability of users in a multi-carrier THz network. Furthermore, for negligible thermal noise, we provide simplified single-integral expressions to compute the outage probability in a multi-carrier network.    Numerical results demonstrate the performance of the proposed user-pairing scheme and validate the accuracy of the derived expressions.
\end{abstract}

\begin{IEEEkeywords}
THz, NOMA, user-pairing, multi-carrier, outage probability, molecular absorption.
\end{IEEEkeywords}

\raggedbottom
\section{Introduction}
Fueled by the emergence of machine-type communications in a variety of wireless applications, the provisioning of massive connectivity becomes instrumental.  On the other hand, accommodating trillions of devices within the extremely congested and limited sub-6GHz spectrum is becoming challenging. In this context, shifting to higher frequency terahertz (THz) communication is under consideration  to obtain the data rates in the order of hundreds of Gbps \cite{rasti2021evolution}.
Also,  non-orthogonal multiple access (NOMA) is becoming popular to support multiple users in the same frequency and time while leveraging on efficient interference cancellation mechanisms.

To date, most of the research focuses on analyzing the outage probabilities (OPs) of users in a \textit{single-carrier sub-6GHz} NOMA network~\cite{8309422,8108691,6868214}.
An exception is \cite{9031560} where the authors  considered   the OP analysis for multi-carrier NOMA in \textcolor{black}{sub-6GHz or radio frequency} (RF) networks. The derived expressions rely on approximations and are in the form of Fox H's function.
Nevertheless, the performance of NOMA in THz networks is not well-understood neither in single-carrier nor in multi-carrier set-up.
\textcolor{black}{Different from RF, the THz transmissions are susceptible to unique challenges such as \textit{molecular absorption noise}\footnote{Molecular absorption noise causes signal loss as the electromagnetic (EM) energy gets partially transformed into internal energy of the molecules.}, \textit{molecular absorption} at different frequencies leading to serious path-loss peaks, and a sophisticated \textit{Beer's Lambert law-based channel} model.}

\textcolor{black}{Recently, a handful of research works considered analyzing the performance of NOMA in THz networks. 
In \cite{SABUJ2020107508}, the authors optimized the transmission powers to maximize energy efficiency in a single carrier THz network. The authors in \cite{9569475} present an energy-efficient cooperative NOMA strategy for multi-user indoor  {multi-input-single-output (MISO) THz network} that assures the minimum required rate for cell-edge users. 
}\textcolor{black}{Lately, preliminary research demonstrated the gains of NOMA over orthogonal multiple access (OMA) in THz networks using computer simulations \cite{9298080}. The authors considered a single THz channel, no fading, random user pairing, and the molecular absorption model was limited to 250 - 450 GHz. 
In \cite{9115278}, the authors optimized power allocation, user clustering, and hybrid precoding to maximize energy efficiency in a THz-NOMA system. In \cite{8824971}, the authors considered maximizing the network sum-rate while optimizing the beamforming weights, sub-array selection, power allocation, and sub-band assignment in a THz-NOMA network subject to the user's quality of service (QoS) requirements. }

\textcolor{black}{
To the best of our knowledge, there is no research work that provides a comprehensive framework for the OP analysis of users in a \textit{single-carrier and multi-carrier THz-NOMA} network and/or develops a low-complexity \textit{user pairing scheme} for the THz-NOMA network with a guaranteed gain over OMA for each individual user.
It is also noteworthy that analyzing the performance of multi-carrier NOMA  in RF networks is straight-forward as all sub-channels experience identical channel statistics. However, in THz networks, each sub-channel experiences a different molecular absorption indicated by its molecular absorption coefficient. Thus the channel statistics on each sub-channel are non-identical. Furthermore, the mathematical structure of the channel propagation model based on  Beer's-Lambert law adds to the challenge. Subsequently, characterizing OP expressions in a multi-carrier THz-NOMA network is challenging.  In the sequel, the main contributions of this paper include:}

\textcolor{black}{
 $\bullet$ We develop a novel low-complexity user pairing scheme in a THz-NOMA network.  The proposed scheme ensures NOMA outperform OMA for each individual user in the NOMA pair and adapts according to the molecular absorption.}
 
$\bullet$ We characterize the exact outage expressions in a single-carrier and multi-carrier THz-NOMA network considering Nakagami-$m$ fading to capture the line-of-sight feature of THz transmissions, and molecular absorption noise. Our expressions can be customized for \textit{various user-pairing schemes} and are applicable to the \textit{entire THz frequency range}.

$\bullet$ For multi-carrier THz-NOMA network,   we propose a moment-generating-function (MGF)-based exact approach and derive simplified  single integral expressions to compute the OP,  as opposed to Fox-H's based expressions in \cite{9031560}.

$\bullet$ Our numerical results validate the accuracy of our derived expressions and demonstrate the performance of the proposed user-pairing scheme compared to the conventional random pairing and nearest-farthest schemes.

\section{System model and Assumptions}
We consider the downlink NOMA transmission of a single-antenna access point (AP) operating at THz frequency. We consider two users\footnote{ \textcolor{black}{Our framework can be generalized for multiple users by making multiple two-user NOMA pairs in orthogonal time/frequency resource blocks. However, the consideration of other NOMA pairs would not effect the performance of users in a specific NOMA pair. Two-user NOMA  has been
standardized as Multi-user Superposition Transmission (MUST) in 3GPP \cite{1111}}.} in each NOMA cluster that are located at distances \(d_1\) and \(d_2\) from the AP such that \(d_1 < d_2\). The users located at \(d_1\) and \(d_2\)  are referred to as user 1 (\(U_1\)) and user 2 (\(U_2\)), respectively. 
\subsubsection{Channel Model} The line-of-sight (LoS) channel power  between the AP and user \(i\) is formulated as follows \cite{sayehvand2020interference}:
 \begin{equation}
 |{h_{L}(d_i )}|^2= \left(\frac{c}{4{\pi} f d_i}\right)^2\ e^{-k(f) d_i}= \zeta  d_i^{-2} e^{-k(f) d_i},
\end{equation}
where the molecular absorption coefficient $k(f)$  is defined as:
\small
\begin{equation}\label{Ka(f)}
{k(f)}\mathrm = \sum_{(i,g)} {\frac{p^2 T_{ \mathrm {sp}}  q^{(i,g)}  {N_A} S^{(i,g)} f \tanh{\left( \frac{h  c  f}{2 k_b T}\right)}} {p_0  V T^2 f^{(i,g)}_c \tanh{\left(\frac{h  c f^{(i,g)}_c}{2 k_b T}\right)}}}   F^{(i,g)}\left(f\right),
\end{equation}
\normalsize
where \(p\) and \(p_0\) indicate the ambient pressure of the transmission medium and the reference pressure, respectively, \(T\) is the temperature of the transmission medium, \(T_{\mathrm{sp}}\) denotes the temperature at standard pressure, \(q^{(i,g)}\) indicates the mixing ratio of the isotopologue \(i\) of gas \(g\), \(N_A\) refers to the Avogadro number, and \(V\) is the gas constant. The line intensity \(S^{(i,g)}\)  defines the strength of the absorption by a specific type of molecules and is directly obtained from the HITRAN database \cite{GORDON20173}. In addition, \(f\) and \(f^{(i,g)}_c\) denote the THz frequency  and the resonant frequency of gas \(g\), respectively, \(c\) is the speed of light, \(h\) is the Planck's constant, and \(k_b\) is  the Boltzmann constant. For the frequency band $f$, we consider the Van Vleck-Weisskopf asymmetric line shape  to evaluate:
\small
\begin{equation}
F^{(i,g)}(f)= \frac{100 \:c \:\alpha^{(i,g)} f} {\pi \:f_c^{(i,g)}} \left(\frac{1}{Y^2+ (\alpha^{(i,g)})^2}+\frac{1}{Z^2+(\alpha^{(i,g)})^2}\right),
\nonumber
\end{equation}
\normalsize
where $Y= f+f_c^{(i,g)}$ and $Z= f-f_c^{(i,g)}$, and the Lorentz half-width is given as follows:
$$\alpha^{(i,g)}= \left( \left(1 - q^{(i,g)} \right) \alpha_{\mathrm {air}}^{(i,g)}+ q^{(i,g)} \alpha_0^{(i,g)}\right) \left(\frac{p}{p_0}\right) \left(\frac{T_0}{T}\right)^{\gamma},$$ where \(T_0\) indicates the reference temperature, the parameters  air half-widths, \(\alpha_{\mathrm{air}}^{(i,g)}\), self-broadened half-widths, \(\alpha_0^{(i,g)}\), and temperature broadening coefficient, \(\gamma\), are obtained directly from the HITRAN database \cite{GORDON20173}. The resonant frequency of gas $g$ at reference pressure \(p_0\) is determined as \( f_c^{(i,g)}=f_{ {c_0}}^{(i,g)} + {\delta}^{{(i,g)}}{(\frac{p}{p_0})} \), where \(\delta^{(i,g)}\) is the linear pressure shift \cite{5995306}.

\subsubsection{SINR - NOMA Model}
The signal-to-interference-plus-noise ratio (SINR) of $U_1$ and $U_2$ with perfect successive-interference cancellation are modeled, respectively, as follows:
\begin{equation} \label{SINR1noma}
    \mathrm{SINR}_{\mathrm{1}}^{(\mathrm{noma})}=\frac{a_1 G_t G_r P |{h_{L}(d_1 )}|^2 \chi_1 }{{N_{1}^\mathrm{(noma)}}},
\end{equation}
\begin{equation}\label{SINR2noma}
    \mathrm{SINR}_{\mathrm{2}}^{(\mathrm{noma})}=\frac{a_2 G_t G_r P |{h_{L}(d_2 )}|^2 \chi_2}{a_1 G_t G_r P |{h_{L}(d_2 )}|^2 \chi_2 +{N_{2}^\mathrm{(noma)}}},
\end{equation}
where $\chi$ is Nakagami-$m$ fading channel, and $G_t$ and $G_r$ are the directional antenna gains of  AP and users, respectively. Beam alignment strategies are assumed that align the main lobes of the users and the THz AP. The noise at the receivers of $U_1$ and $U_2$ comprises of thermal noise $N_0$ and molecular absorption noise as defined, respectively, below:
\begin{align}
 &N_{1}^\mathrm{(noma)}= N_0 + a_1 G_t G_r\zeta Pd_1^{-2} (1-e^{-k(f)d_1}) \chi_1,
 \end{align}
 \begin{align}
    {N_{2}^\mathrm{(noma)}} &= N_0 + (a_1  + a_2) G_t G_r \zeta Pd_2^{-2}(1-e^{-k(f)d_2}) \chi_2,
   \nonumber \\
     &= N_0 + G_t G_r \zeta Pd_2^{-2} (1-e^{-k(f)d_2}) \chi_2,
\end{align}
where  \(a_1\) and \(a_2\) depict the fraction of the AP transmit power allocated for \(U_1\) and \(U_2\), respectively, such that  \(a_1 + a_2 = 1\). Also, \(P\) denotes the total transmit power budget of the AP.  
The spectral efficiency of \(U_1\) and \(U_2\) (in bps/Hz) is computed for a duration of time $\hat T$ as follows: 
$$
   \mathrm{C}_{\mathrm{i}}^{(\mathrm{noma})}= \hat T \log_2(1+ \mathrm{SINR}^{(\mathrm{noma})}_{i}), \quad \forall i =\{1,2\}.
$$
\subsubsection{SINR - OMA Model}
The SINR for \(U_1\) and \(U_2\) in  OMA, where each user receives its transmission for a predefined duration of \(\hat T/2\), is modeled as follows:
\begin{equation} \label{SINRioma}
      \mathrm{SINR}_{{i}}^{(\mathrm{oma})}=\frac{G_t G_r P |{h_{L}(d_i )}|^2 \chi_i}{N^{(\mathrm{oma})}_i}, \quad \forall i =\{1,2\},
\end{equation}
where  
$
N_{i}^{(\mathrm{oma})}=N_0+G_t G_r\zeta P d_i^{-2}\left(1-e^{-k(f)d_i}\right) \chi_i.
$
Thus, the spectral efficiency of \(U_1\) and \(U_2\) can be computed as $$
\mathrm{C}_{{i}}^{(\mathrm{oma})}={\frac{\hat T}{2}\log_2}{\left(1+ \mathrm{SINR}_{i}^{(\mathrm{oma})}\right)}, \quad \forall i =\{1,2\}.
$$
Without loss of generality, the duration $\hat T$ is taken as unity.

\section{Outage Analysis: single-Carrier THz-NOMA}
In this section, we first describe the proposed user-pairing scheme along with the two benchmark user-pairing schemes, describe the distance distributions of the users, and present a framework to calculate the outage of users for the proposed and benchmark user-pairing schemes.

For bench-marking purposes, we consider a random  and nearest-farthest user pairing schemes. \textcolor{black}{In the random pairing scheme, we pick only two users randomly with independent and identically distributed distances $r_1$ and $r_2$ from AP. The near  and far user's distance can thus be defined as  $d_1 = \mathrm{min} (r_1,r_2)$ and $d_2 = \mathrm{max} (r_1,r_2)$, respectively.}  Therefore, the PDF and CDF of the distances of  \(U_1\) and \(U_2\) are given, respectively, as:
\small
\begin{equation}\label{CDFRandom}
     f_{d_1}(d_1)=2f_r\left(d_1\right)(1-F_r\left(d_1\right)), \:\: F_{d_1}(d_1)=1-[1-F_r\left(d_1\right)]^2,
\end{equation}
\begin{equation} \label{CDFRandomFar}
     f_{d_2}(d_{2})=2f_r\left(d_2\right)F_r\left(d_2\right), \quad  F_{d_2}\left(d_2\right)=[F_r\left(d_2\right)]^2,
\end{equation}
\normalsize
On the other hand, in the nearest-farthest scheme, we select two users out of $N$ users with minimum and maximum distances from the AP. \textcolor{black}{The near  and far user's distance can thus be defined as  $d_1 = \mathrm{min} (r_1,r_2, \cdots, r_N)$ and $d_2 = \mathrm{max} (r_1,r_2, \cdots, r_N)$, respectively.} 
The PDF and CDF of $d_1$ and $d_2$  can thus be given, respectively, as follows:
\small
\begin{equation}\label{CDF_NF}
     f_{d_1}(d_1)=N[1-F_r\left(d_1\right)]^{N-1}  f_r\left(d_1\right), 
     F_{d_1}(d_1)=1-\left[1-F_r\left(d_1\right)\right]^N,
\end{equation}
\begin{equation}\label{CDF_NF_Far}
 f_{d_2}\left(d_2\right)=N\left[F_r\left(d_2\right)\right]^{N-1} f_r\left(d_2\right), F_{d_2}\left(d_2\right)=\left[F_r\left(d_2\right)\right]^N,
\end{equation}\normalsize
where the PDF and CDF of $r$ are given, respectively, as 
$f_r\left(r\right)=\frac{2r}{R^2}, \mathrm{and} \: F_r\left(r\right)=\frac{r^2}{R^2},
$ since all users are uniformly distributed in a circular region of radius $R$.
\begin{figure}
\begin{center}
\includegraphics[totalheight=6cm,width=\linewidth]{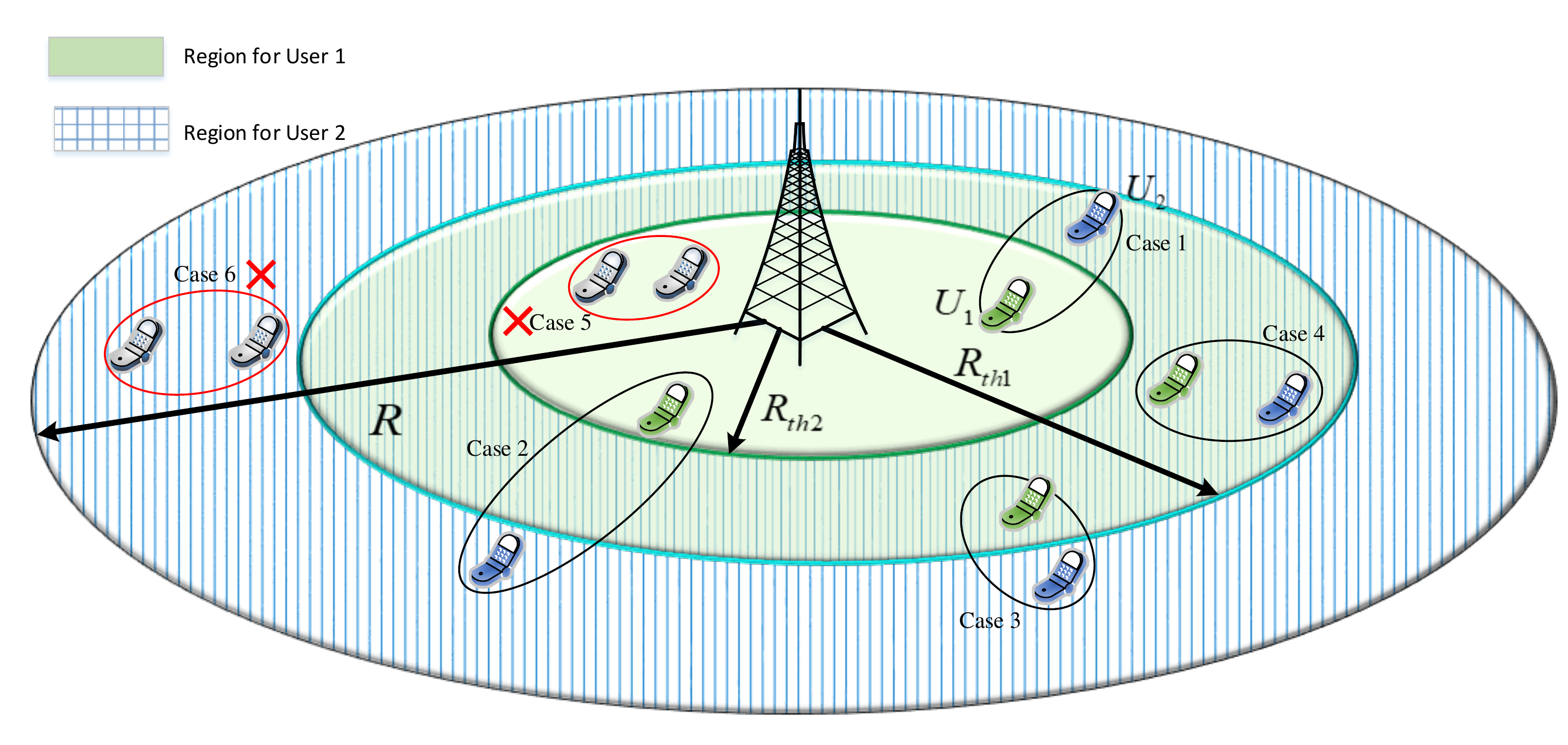}
\caption{An illustration of the proposed user pairing scheme.}
\label{Fig1}
\end{center}
\end{figure}

\subsection{Proposed User Grouping Scheme}

\textcolor{black}{ Different from the conventional approach where user-pairing is typically performed to maximize the sum-rate performance of users in a NOMA pair, the proposed user scheme selects  only those users  who can benefit from NOMA transmissions compared to OMA transmission.  The  scheme enables fairness among users by allowing each of them to join a NOMA pair if and only if they have a benefit over OMA. 
}
{To ensure the gains of NOMA over OMA for each user, we derive the necessary condition to ensure \({C}_{{i}}^\mathrm{(noma)} > {C}_{{i}}^\mathrm{(oma)},\:\forall i=\{1,2\}\) and  \(d_1<d_2\) 
as shown in the  following.}
\begin{lemma}[Proposed User Pairing Scheme] 
\textcolor{black}{The proposed pairing scheme is a sufficient condition for NOMA to outperform OMA, where the near  user should be within the distance \( R_\mathrm{th1}\) and the far user should be located beyond the distance \(R_\mathrm{th2} \), such that $d_1<d_2$, where:}
\begin{equation}\label{eq:13}
  R_{\mathrm{th1}}= \frac{1}{\mathrm{} {k(f)}} \ {\mathrm{ln }  \left(\frac{1-a_1}{1-2a_1}\right)},
 \end{equation}
 \begin{equation}\label{eq:14}
  R_{\mathrm{th2}}= \frac{1}{\mathrm{} {k(f)}} \ {\mathrm{ln } \left(\frac{{a_1}^2}{1-2a_1}+1  \right)}.
 \end{equation}
 Comparing (\ref{eq:13}) to (\ref{eq:14}), it is straight-foward to verify that $R_{\mathrm{th1}}> R_{\mathrm{th2}}$ always, since $a_1+a_2=1$.
 \end{lemma}
 \begin{proof}
  \color{black}{
 Starting from the condition ${C}_{{i}}^\mathrm{(noma)} > {C}_{{i}}^\mathrm{(oma)},\:\forall i=\{1,2\}$ and taking $\hat{T} =1$ without loss of generality, we have:
 \small
\begin{equation} \label{NOMAivsOMAi}
    \log_2{\left(1+ \mathrm{SINR}_{{i}}^{(\mathrm{noma})}\right)} > {0.5\log}_2{\left(1+\mathrm{SINR}_{{i}}^{(\mathrm{oma})}\right)} ,\:\forall i=\{1,2\}.
\end{equation}
\normalsize
Now, we substitute \eqref{SINR1noma} and  \eqref{SINRioma} in \eqref{NOMAivsOMAi}  for near user, and  \eqref{SINR2noma} and \eqref{SINRioma} in \eqref{NOMAivsOMAi} for far user. After basic algebraic manipulations, we note that $d_1 < R_{\mathrm{th}1}$ and $d_2>R_{\mathrm{th}2}$ are necessary conditions to guarantee the gains of NOMA over OMA for each individual user, where the value of $R_{\mathrm{th}1}$ and $R_{\mathrm{th}2}$  can be written as in \eqref{eq:13} and \eqref{eq:14}, respectively.}
 \end{proof}

 Evidently, as can also be seen in Fig. 1, four user-pairing cases are possible to guarantee the gains of NOMA over OMA for each individual user. The three  regions can vary as a function of $k(f) $ and the allocated powers $a_1$ and $ a_2$, e.g., { increasing $a_1$  will increase both $R_{\mathrm{th1}}$ and $R_{\mathrm{th2}}$. However, increasing $k(f)$ will decrease both $R_{\mathrm{th1}}$ and $R_{\mathrm{th2}}$.} This is different from RF NOMA, wherein  there is only one case possible due to only one threshold distance (or two regions).
 \begin{lemma}
The PDF and CDF of $d_1$ and $d_2$  in the  proposed  scheme are given, respectively, as follows:
\begin{equation}\label{CDFProposed}
    f_{d_1}\left(d_1\right)=\frac{2d_1}{{R^2_\mathrm{th1}}},F_{d_1}\left(d_1\right)=\frac{{d^2_1}}{{R^2_\mathrm{th1}}},
\end{equation}
\begin{equation}\label{CDFProposedFar}
f_{d_2}\left(d_2\right)=\frac{2d_2}{R^2-R_\mathrm{th2}^2},
F_{d_2}\left(d_2\right)=\frac{d_2^2-R_\mathrm{th2}^2}{R^2-R_\mathrm{th2}^2}.
\end{equation}
\end{lemma}

\begin{proof}
\color{black}{
Since all users are uniformly distributed in a circular region of radius $R$,  thus the PDF and CDF of a user from the BS  $r$ are given, respectively, as 
$f_r\left(r\right)=\frac{2r}{R^2} \mathrm{and} \: F_r\left(r\right)=\frac{r^2}{R^2}$. In the proposed scheme, the range of the near user is $ d_{1}\in\ [0,R_{\mathrm{th}1}]$. Subsequently, $d_1$ follows the truncated distribution of $r$ and its CDF can be calculated as follows:
$$F_{d_1}\left(d_1\right)=\frac{F_r\left(d_1\right) - F_r\left(0\right)\ }{F_r\left(R_\mathrm{th1}\right) - F_r\left(0\right)}
 =\frac{\frac{{d_{1}}^2}{{R}^2} -\frac{({0)}^2}{{R}^2}}{\frac{R^2_\mathrm{th1}}{{R}^2}-\ \frac{({0)}^2}{{R}^2}\ }
=\frac{{d^2_1}}{{R^2_\mathrm{th1}}},
$$
Now, we can calculate the PDF by taking the derivative of   the CDF as shown in \textbf{Lemma~2}. 
Similarly,  the far user is located in the range
$ d_{2}\in\ [R_{\mathrm{th}2}, R] $ and the CDF of $d_2$ is given as:
$$F_{d_2}\left(d_2\right)=\frac{F_r\left(d_2\right) - F_r\left(R_\mathrm{th2}\right)\ }{F_r\left(R\right) - F_r\left(R_\mathrm{th2}\right)}
=\frac{d_{2}^2-R_\mathrm{th2}^2}{R^2-R_\mathrm{th2}^2},
$$
Now, we can calculate the PDF of $d_2$ by taking the derivative of the CDF as shown in \textbf{Lemma~2}. 
}
\end{proof}
\begin{corollary}[Enhanced Proposed Scheme]
The proposed scheme can be enhanced further in terms of spectral efficiency by selecting the nearest user as $U_1$.  In this case,  $f_{d_1}(d_1)$ and $F_{d_1}(d_1)$ can be given as in \eqref{CDF_NF}.
 \end{corollary}
\begin{proof}
\color{black}{
To enhance the proposed scheme, we consider that the near user is the user with minimum distance, thus,  $f_{d_1}(d_1)$ and $F_{d_1}(d_1)$ can be given as in \eqref{CDF_NF}.
On the other hand, the far user is located beyond \(R_\mathrm{th2} \). So the range of the far user will be as 
$ d_{2}\in\ [R_\mathrm{th2},R] $
, then in this case, $f_{d_2}(d_2)$ and $F_{d_2}(d_2)$ can be given as in \eqref{CDFProposedFar}.
}
\end{proof}
\subsection{OP Analysis}
The  OP is defined as the probability that user $i$ does not achieve its target-spectral efficiency $\tau_i$, i.e.,
\begin{equation}
  {\mathcal{O}^{({\cdot})}_i}=\Pr{\left(\mathrm{C}^{({\cdot})}_{\mathrm{i}}\le\ \mathrm \tau_i\right)},
\end{equation}
where ${({\cdot})}$ denotes NOMA or OMA transmission. Now we formally derive the  OP of users in  NOMA and OMA separately.
\subsubsection{NOMA (Near User)}
The  OP of the near user in the downlink  NOMA mode is given as:
\begin{align} \label{outnear1}
    \mathcal{O}^\mathrm{(noma)}_1 
    & {=} \mathrm {Pr}{\left(\frac{B_1 d_1^{-2}{e^{{-k}\left(f\right)d_1} } { \chi}_1}{N_0+B_1d_1^{-2}\left(1-e^{-k\left(f\right)d_1}\right){ \chi}_1} \le\ { y}_1\right)}\nonumber
    \\&{=}\int_0^R\frac{\gamma\left[m, \frac{ {y_1} N_0 d_1^2}{{\Theta B}_1\left(e^{-k\left(f\right)d_1}\left(1+ y_1 \right)- {y_1} \right)}\right]}{\mathrm{\Gamma}(m)} f_{d_1}(d_1) \; d d_1,
\end{align}
where $\gamma(\cdot)$ is the lower incomplete Gamma function, $\Gamma(\cdot)$ is the complete Gamma function, $m$ is the fading severity,  $\Theta$ is the fading power, \(y_1 = 2^{\tau_1}-1\) and \({B_1{=a}_1G}_{t}G_{r}P\zeta\).
\textcolor{black}{Since the thermal noise \(N_0\) is negligible compared to the molecular absorption noise in THz networks}, \eqref{outnear1} is simplified as:
\small
\begin{equation}\label{OutNOMA}
     \hat{\mathcal{O}}^\mathrm{(noma)}_1 
        =1-\mathrm {Pr}{\left( d_1\le \frac{\ln{(\frac{1+y_1}{{ y}_1})}}{k\left(f\right)}\right)}
       = 1-F_{d_1}\left( \frac{\ln{(\frac{1+y_1}{{ y}_1})}}{k\left(f\right)}\right).
\end{equation}
\normalsize
Now substituting \eqref{CDFRandom} in \eqref{OutNOMA}  or \eqref{outnear1} for random scheme, \eqref{CDF_NF} in \eqref{OutNOMA} or \eqref{outnear1} for nearest-farthest scheme, and \eqref{CDFProposed} in \eqref{OutNOMA}  or \eqref{outnear1} for proposed scheme, gives us the  outage expressions. The  outage of $U_1$ in the proposed scheme is given as:
\begin{equation}\label{OutNOMA2}
    \hat{\mathcal{O}}^\mathrm{(noma)}_1 = 1-{{\left(\frac{\ln{(\frac{1+y_1}{{ y}_1})}}{R_\mathrm{th1}\:k\left(f\right)}\right)^2}}.
\end{equation}
\subsubsection{OMA (Near User)}
The  OPs with and without noise can be given by replacing $B_1$ with \( {A=G}_tG_rP\zeta\) and $y_1$ with \({{x_1}=2}^{2{\tau_1}}- 1\) in \eqref{outnear1}, \eqref{OutNOMA}, and \eqref{OutNOMA2} for all schemes.

\subsubsection{NOMA (Far User)}
The  OP of the far user in the downlink  NOMA mode is formulated as follows:
 \small
 \begin{align}\label{infar1}
    &\mathcal{O}^\mathrm{(noma)}_2  
       = \mathrm {Pr}{\left (\frac{B_2 |{h_{L}(d_2 )}|^2 \chi_2}{B_1 |{h_{L}(d_2 )}|^2 \chi_2 +{N_{2}^{(\mathrm{noma})}}} \le y_2\right)}
       \nonumber  \\&{=} \int_0^R \frac{\gamma\left[m,\ \frac{ y_2 N_0 d_2^2}{_\Theta \left(e^{{-k}\left(f\right)d_2}({B_2}\ {+\ y_2 B}_2)-y_2\ ({B_1+B}_2)\ \ \right)}\right]}{\Gamma(m)} f_{d_2}(d_2) d d_2,
\end{align}
\normalsize
  where, $B_2=(a_1+a_2) G_tG_rP\zeta$. {since \(N_0\) is negligible compared to molecular absorption noise},  \eqref{infar1} is simplified as:
\begin{equation}\label{InNOMA}
    \begin{split}
       \hat{\mathcal{O}}^\mathrm{(noma)}_2
               = 1-F_{d_2}\left(\frac{\ln{\left[\frac{a_2+ a_2y_2}{y_2}\right]}}{k\left(f\right)}\right).
    \end{split}
\end{equation}
Now substituting \eqref{CDFRandomFar} in \eqref{InNOMA}  for random scheme, \eqref{CDF_NF_Far} in \eqref{InNOMA} for nearest-farthest, and \eqref{CDFProposedFar} in \eqref{InNOMA} for proposed scheme, gives us the  outage expressions. The outage of the far user in the proposed scheme is given as:
\begin{equation} \label{inNOMA2}
 \hat{\mathcal{O}}^\mathrm{(noma)}_2=1-
\frac{\left(\frac{\ln{\left[\frac{a_2+ a_2y_2}{y_2}\right]}}{k\left(f\right)}\right)^2-R_\mathrm{th2}^2}{R^2-R_\mathrm{th2}^2}.
\end{equation}
\normalsize
\subsubsection{OMA (Far User)}
The  OPs  can be given by replacing  $d_1$ with $d_2$, $B_1$  with \({A=G}_tG_rP\zeta\) and $y_2$ with \({{x_2}=2}^{2{\tau_2}}- 1\) in \eqref{outnear1}, \eqref{OutNOMA}, and \eqref{OutNOMA2} for all schemes.

\section{Outage Analysis: Multi-Carrier THz-NOMA}
In this section, we present moment-generating function (MGF)-based approach  to derive tractable  outage expressions of the users in a multi-carrier THz-NOMA network. Both the near and far users will get $N$ subcarriers allocated. The spectral efficiency of $U_1$ and $U_2$ is given, respectively, as:
\begin{equation}\label{SINR_MC_user1}
    C_{1}^\mathrm{(noma)}=\sum_{n=1}^{N}\log_2{\left(1+\frac{{B_1{\chi}_{1,n} d_1^{-2}}{{e^{{-k}\left(f_n\right)d_1}}}}{N_{1,n}^\mathrm{(noma)}}\right)},
\end{equation}
\begin{equation}\label{SINR_MC_user2}
 C_{2}^\mathrm{(noma)} =\sum_{n=1}^{N}\log_2{\left(1+\frac{B_2 \chi_{2,n} d_2^{-2} e^{-k(f_n)d_2}}{B_{1}\chi_{2,n}d_2^{-2} e^{-k(f_n)d_2}+N_{2,n}^{(\mathrm{noma})}}\right)}.
\end{equation}
The  OP can thus be formulated as follows:
\begin{align} \label{Eq:26}
    \mathcal{O}_i^{(\mathrm{noma})} &= \mathrm{Pr}\left(C_i^{\mathrm{(noma)}}< \tau_i\right)
 = \mathrm{Pr} {\left(\sum_{n=1}^{N} \mathrm{log_2}{\left({W}_{i,n}\right)} < \tau_i\right)},
\nonumber\\& = \mathrm{Pr} {\left(X_i=\sum_{n=1}^{N} \mathrm{\ln}{\left({W}_{i,n}\right)} < \tau_i\ln(2)\right)}.
   \end{align}
\textcolor{black}{Note that the MGF is a useful tool to deal with the sum of random variables. For instance, characterizing the PDF becomes analytically intractable for a sum of random variables due to  multiple  convolutions required. On the other hand, the MGF of a sum of random variables can be  derived by deriving the product of the MGF of all random variables. Therefore, to analyze the OP of users in a multi-carrier THz-NOMA network, we resort to an MGF-based approach. Our methodology is as follows: \textbf{(i)} we first derive the PDF of $W_{i,n}$ conditional on $d_1$ for near user and $d_2$ for far user, \textbf{(ii)} derive the conditional MGF of $X_{i,n} = \mathrm{ln} \, W_{i,n}$,  \textbf{(iii)} compute the conditional cumulative MGF of $X_i=\sum_{n=1}^N \mathrm{ln} \, W_{i,n}$, i.e., $M_{X_i|d_i}=\prod_{n=1}^N M_{X_{i,n}|d_i}$, and \textbf{(iv)} substitute in the Gil-Pelaez inversion lemma to compute the  OP as follows:  }
\begin{align}
\label{eq:Twave_cov_only}
&\mathcal{O}_i^{(\mathrm{noma})} =  
\frac{1}{2}+\frac{1}{\pi}\int_{0}^{\infty}\frac{\mathrm{Im}[M_{X_{i}}(s) e^{j \omega \tau_i \ln(2)}]}{\omega} d\omega,
\end{align}
where $s=j\omega$, $M_{X_i}(s)=\mathbb{E}_{d_i}[\prod_{n=1}^N, M_{{X_{i,n}|d_i}}]$, and
\begin{align}\label{MGF}
M_{X_{i,n}|d_i}(s) &=\textcolor{black}{ E[e^{{-s} \ln{\left({W}_{i,n}\right)}}|d_i ]= E [{{W_{i,n}}^{-s}} |d_i]}, 
\nonumber\\&= \int_{-\infty}^{\infty}{{W}_{i,n}}^{-s}{f}_{{W}_{i,n}}\left({W}_{i,n}\right) {d}{{W}_{i,n}} .
\end{align}
From \eqref{SINR_MC_user1} and \eqref{SINR_MC_user2}, we have the following: $$\chi_{1,n}= g_1(W_{1,n})=\frac{N_0(W_{1,n}-1)}{W_{1,n}\left({e}^{-k\left(f_n\right)d_1}-1\right)B_1 d_1^{-2} + B_1 d_1^{-2}},$$
$${\chi}_{2,n}=g_2(W_{2,n})=\frac{N_0-N_0{W}_{2,n}}{{{B}_2d}_2^{-2}({ W}_{2,n}-{a_2e}^{{-k}\left(f_n\right)d_2}{{\ W}_{2,n}}-1)}.$$
Since $\chi$ is Gamma distributed, we apply random variable transformation to get the PDF of $W_{i,n}$: 
\small
\begin{align}\label{PDF_MC_user1_and_User_2}
     f_{W_{i,n}}\left(W_{i,n}\right)=\frac{B_{i} d_i^{-2} A_i}{ N_{0} { e}^{-k\left(f_n\right)d_i}}\left.\left(\frac{\chi_{i,n}^{m-1} e^{\frac{-{\chi}_{i,n} }{\Theta}}}{\Theta^m\Gamma\left(m\right)}\right)\right|_{{{\chi}_{i,n}=g_1(W_{i,n})}},
\end{align}\normalsize
where $A_1= \left({W}_{1,n}\left({e}^{-k\left(f_n\right)d_1}-1 \right)+1\right)^2$ and $A_2= \left({W}_{2,n}\left(a_2{e}^{-k\left(f_n\right)d_2}-1 \right)+1\right)^2/a_2$.
Finally, we can obtain  $M_{X_{1,n}}(s)$ for $U_1$ and $M_{X_{2,n}}(s)$ for $U_2$ by substituting \eqref{PDF_MC_user1_and_User_2} for $i=1$ in \eqref{MGF} and $i=2$ in \eqref{MGF}, respectively.

\begin{lemma}[OP in Multi-carrier Network]\label{lem3}
The outage can be computed of user $i$ in a multi-carrier THz-NOMA network with the proposed scheme can be derived as follows: 
\begin{align}\label{lemma3}
\mathcal{O}_i^{(\mathrm{noma})}&= \mathrm{Pr} {\left(X_i=\sum_{n=1}^{N} \mathrm{\ln}{\left({W}_{i,n}\right)} < \tau_i\ln(2)\right)}
\nonumber \\&=\int_0^R \mathbb{U}(X_i-\tau_i \mathrm{ln} (2)) f_{d_i}(d_i)   dd_i .
\end{align}
Now substituting \eqref{CDFRandom}  in \eqref{lemma3} for random scheme and \eqref{CDF_NF}  in \eqref{lemma3} for nearest-farthest scheme gives the respective  OPs.
\end{lemma}
\begin{proof}
\color{black}{
From \eqref{SINR_MC_user1} and \eqref{SINR_MC_user2}, we can write 
 \begin{align}
     &{W}_{1,n}=1+\frac{{B_1{\chi}_{1,n} d_1^{-2}}{{e^{{-k}\left(f_n\right)d_1}}}}{N_{1,n}^\mathrm{(noma)}}\\&
 {W}_{2,n}=1+\frac{B_2 \chi_{2,n} d_2^{-2} e^{-k(f_n)d_2}}{B_{1}\chi_{2,n}d_2^{-2} e^{-k(f_n)d_2}+N_{2,n}^{(\mathrm{noma})}}
 \end{align}
 For enhanced tractability, we ignore the thermal noise and obtain the following simplified results after algebraic manipulations for near and far users, respectively.
$${W}_{1,n}=\frac{1}{1-e^{-k\left(f_n\right)d_1}} \quad \mathrm{or} \quad  d_1=\frac{1}{k\left(f_n\right)}{\ln{\left[\frac{{W}_{1,n} }{{W}_{1,n}-1}\right]}},$$
$${W}_{2,n}=\frac{1}{ 1-{a_2e}^{-k\left(f_n\right)d_2}} \quad \mathrm{or} \quad  d_2=\frac{1}{k\left(f_n\right)}{\ln{\left[\frac{a_2{W}_{2,n}}{{ W}_{2,n}-1}\right]}}.$$
 From \eqref{Eq:26} and using the fact that $X_i$ is a constant conditional on $d_i$ and the CDF of a constant is a unit-step function, we got to \eqref{lemma3} as in \textbf{Lemma~3}.
}
\end{proof}
From \eqref{eq:13}, we note that the $R_{\mathrm{th}1}$ and $R_{\mathrm{th}2}$ will vary for each carrier, since they depend on the frequency-dependent absorption coefficient~$k(f_n)$ which is computed using \eqref{Ka(f)}. Thus, each carrier should be allocated to a user selected from a different region based on the proposed scheme.
 Therefore, we extended our scheme for multi-carrier networks by choosing a user from a region that is valid for all subcarriers. 
\begin{corollary}
We consider  $R_{\mathrm{th}1}^\mathrm{min} =  \mathrm{min}(R_\mathrm{th_n}), \forall n \in \{1,2,\cdots,N\}$ and $R_{\mathrm{th}2}^\mathrm{max} =  \mathrm{max}(R_\mathrm{th_n}), \forall n \in \{1,2,\cdots,N\}$ for near user and far user, respectively.
Subsequently, after applying $R_\mathrm{th1}^\mathrm{min}$ and $R_\mathrm{th2}^\mathrm{max}$, the PDF  of the distance of near user and far user can be given, respectively, as follows:
\begin{align}\label{final}
 f_{d_1}\left(d_1\right)=\frac{2d_1}{{(R_\mathrm{th1}^\mathrm{min})^2}}, f_{d_2}\left(d_2\right)=\frac{2d_2}{R^2-{(R_\mathrm{th2}^\mathrm{max})^2}} .  
\end{align}
\end{corollary}
\begin{proof}
\color{black}{ The near user will need to be located inside $R_{\mathrm{th}1}$.
However, each subcarrier will observe a different molecular absorption coefficient resulting in a different threshold distance  $R_{\mathrm{th}1}$ at each subcarrier.  Choosing the smallest threshold distance will not violate the threshold requirement of all other subcarriers, thus we consider  $R_{\mathrm{th}1}^\mathrm{min} =  \mathrm{min}(R_\mathrm{th_n}), \forall n \in \{1,2,\cdots,N\}$ for near user. The range of $d_1$ becomes 
$d_{1}\in\ [0,R_{\mathrm{th}1}^\mathrm{min}]$, thus replacing $R_{\mathrm{th}1}$ with $R_{\mathrm{th}1}^\mathrm{min}$ with  in Lemma~2, we can obtain the result in Corollary~2.
Similarly, the far user will need to be located outside $R_{\mathrm{th}2}$. Therefore, choosing the maximum threshold distance will not violate the threshold requirement of all other subcarriers, i.e.,  $R_{\mathrm{th}2}^\mathrm{max} =  \mathrm{max}(R_\mathrm{th_n}), \forall n \in \{1,2,\cdots,N\}$  for far user. The range of $d_2$ becomes 
$d_{2}\in\ [0,R_{\mathrm{th}1}^\mathrm{min}]$, thus replacing $R_{\mathrm{th}2}$ with $R_{\mathrm{th}2}^\mathrm{max}$  in Lemma~2, we can obtain the result in Corollary~2.
}
\end{proof}
Now substituting \eqref{final}  in \eqref{lemma3} for the proposed scheme gives  the outage expressions of users in our proposed scheme.

\begin{table}
  \centering
  \caption{Simulation Parameters for Calculating $k(f)$ in \eqref{Ka(f)} }
    \label{MOLECULAR ABSORPTION PARAMETERS}
    \begin{tabular}{| p{1.375cm} | p{2.4cm} | p{0.8cm} |p{3.05cm} |}
     \hline 
     \textbf{Symbol}  &\textbf{Value}   &\textbf{Symbol}& \textbf{Value}
    \\\hline
    $p_0, p$  &1 atm, 1 atm & $q^{(i,g)}$&0.05 [\%]\\
    $T_0,T$ & 296~K, 396~K  &$k_b$ &1.3806$\times 10^{-23}$ J/K\\ 
    $f_{ {c_0}}^{(i,g)}$ &276 Hz & $T_{\mathrm{sp}}$ &273.15 K\\
    $\gamma$ &0.83 & $N_A$ & 6.0221 $\times 10^{23}$\\   
    $S^{(i,g)}$ &2.66$^{-25}$Hz-m$^2/$mol & h &6.6262$\times 10 ^{-34}$ J s\\
    $\alpha_0^{(i,g)}, \alpha_{\mathrm{air}}^{(i,g)}$ &0.916Hz, 0.1117Hz & c&2.9979 $\times 10^8$ m/s \\
    ${\delta}^{{(i,g)}}$ &0.0251 Hz & V &8.2051$\times 10^{-5}$m$^3$atm/K/mol\\
    \hline
  \end{tabular}
\end{table}

\section{Numerical Results and Discussions}
In this section, we compare the performance of $U_1$ and $U_2$ in a single-carrier and multi-carrier THz-NOMA network, considering a variety of user-pairing schemes  (i) random scheme, (ii) proposed scheme, (iii) enhanced scheme, and (iv) near-far user pairing schemes.

Unless stated otherwise, the  parameters  are listed herein. We consider 300 users are uniformly distributed  in a circular disc of radius 60~m.  The antenna gains $G_{t}$ and $G_{r}$ are set as 20 dB. The AP transmit power  is  1W and the power allocation coefficients  $a_1 =0.33$ and $a_2=1-a_1$. \textcolor{black}{Note that our framework is general for any arbitrary value of $a_1$ in the range $0 \leq a_1 < 0.5$.}  Nakagami-$m$ fading parameter is set as 2 and $\Omega =1$. In multi-carrier NOMA, we consider six subcarriers, where each subcarrier has the same transmission bandwidth, with the frequencies [0.85, 0.9,  0.95,  1.0,   1.05,  1.1] THz and their respective  $k(f)=$[0.0357, 0.04, 0.0446, 0.0494, 0.0545,   0.0598] m$^{-1}$ computed using \eqref{Ka(f)} considering water vapour molecules. We list the numerical values of parameters in \textbf{Table~\ref{MOLECULAR ABSORPTION PARAMETERS}} which are taken from \cite{9216613}. 

\begin{figure}
\begin{center}
\includegraphics[totalheight=5.5cm,width=\linewidth]{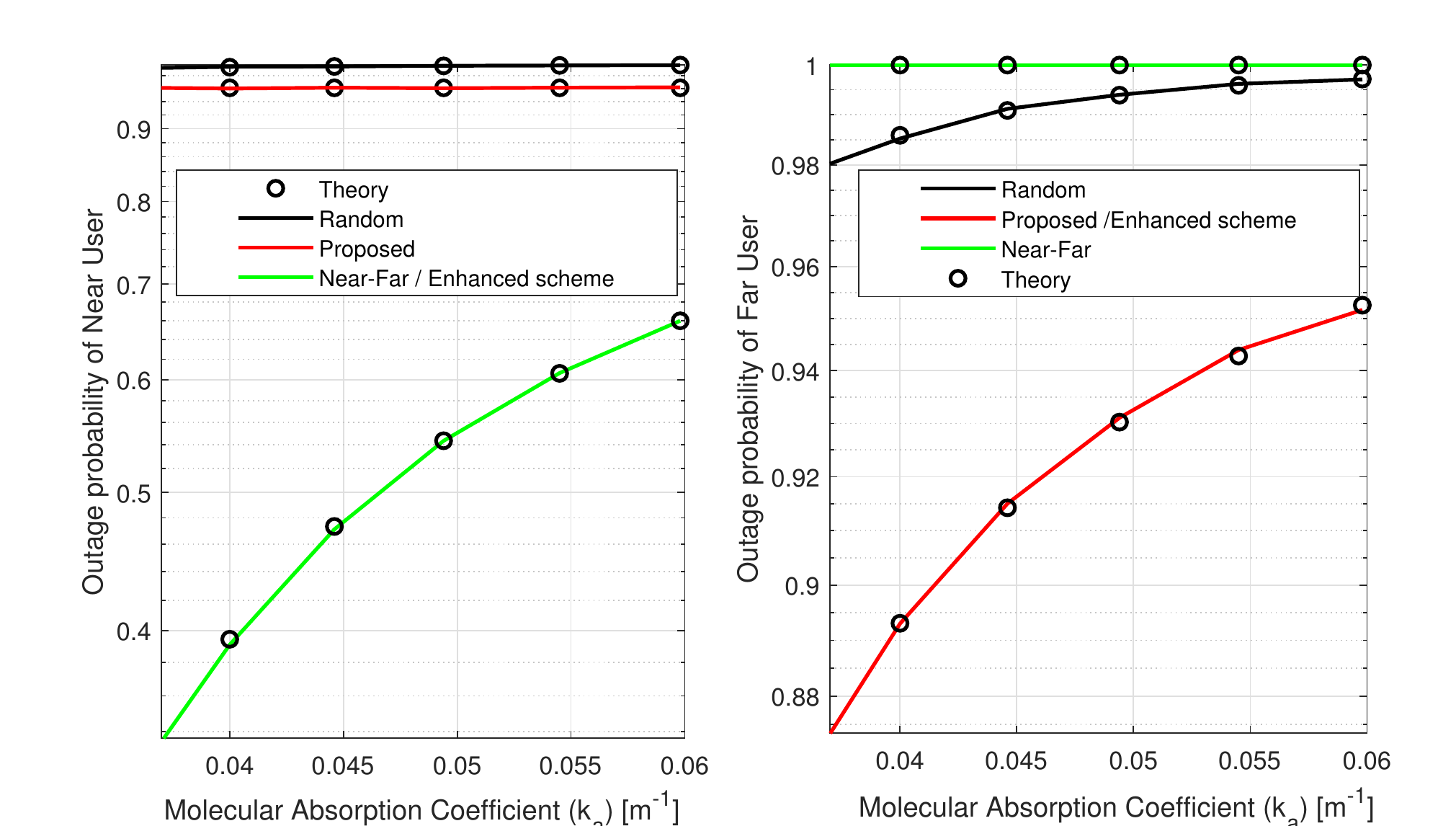}
\caption{Outage performance of near and far users as a function of the molecular absorption coefficient in THz spectrum considering a single-carrier network, $R$ = 60 m, $\tau_1$ = 3 bps/Hz, and $\tau_2$ = 0.5 bps/Hz.}
\label{Fig2}
\end{center}
\vspace{-5 mm}
\end{figure}

Fig.~2  depicts the OP of  $U_1$ and $U_2$ for different values of $k(f)$ and validates \eqref{outnear1} and \eqref{infar1} through Monte-Carlo simulations.  It is shown that the values obtained through derived expressions (shown in circles), exactly match those obtained through simulations (shown by lines). The enhanced scheme (Corollary~1) significantly outperforms the random and nearest-farthest schemes. \textcolor{black}{It is interesting to note that our scheme adapts user selection based on the molecular absorption coefficient $k(f)$. That is, \(R_\mathrm{th1}\) and \( R_\mathrm{th2}\) reduce with the increase in $k(f)$  as can be seen from (\ref{eq:13}) and (\ref{eq:14}), respectively. Thus, closeby users are selected  to combat the effect of increased molecular absorption. Finally, it can be seen that  the OP increases with the increase in frequency and $k(f)$  due to increased molecular absorption; therefore, a lower THz frequency with lower $k(f)$ is preferred.}

\begin{figure}
\begin{center}
\includegraphics[totalheight=5cm,width=\linewidth]{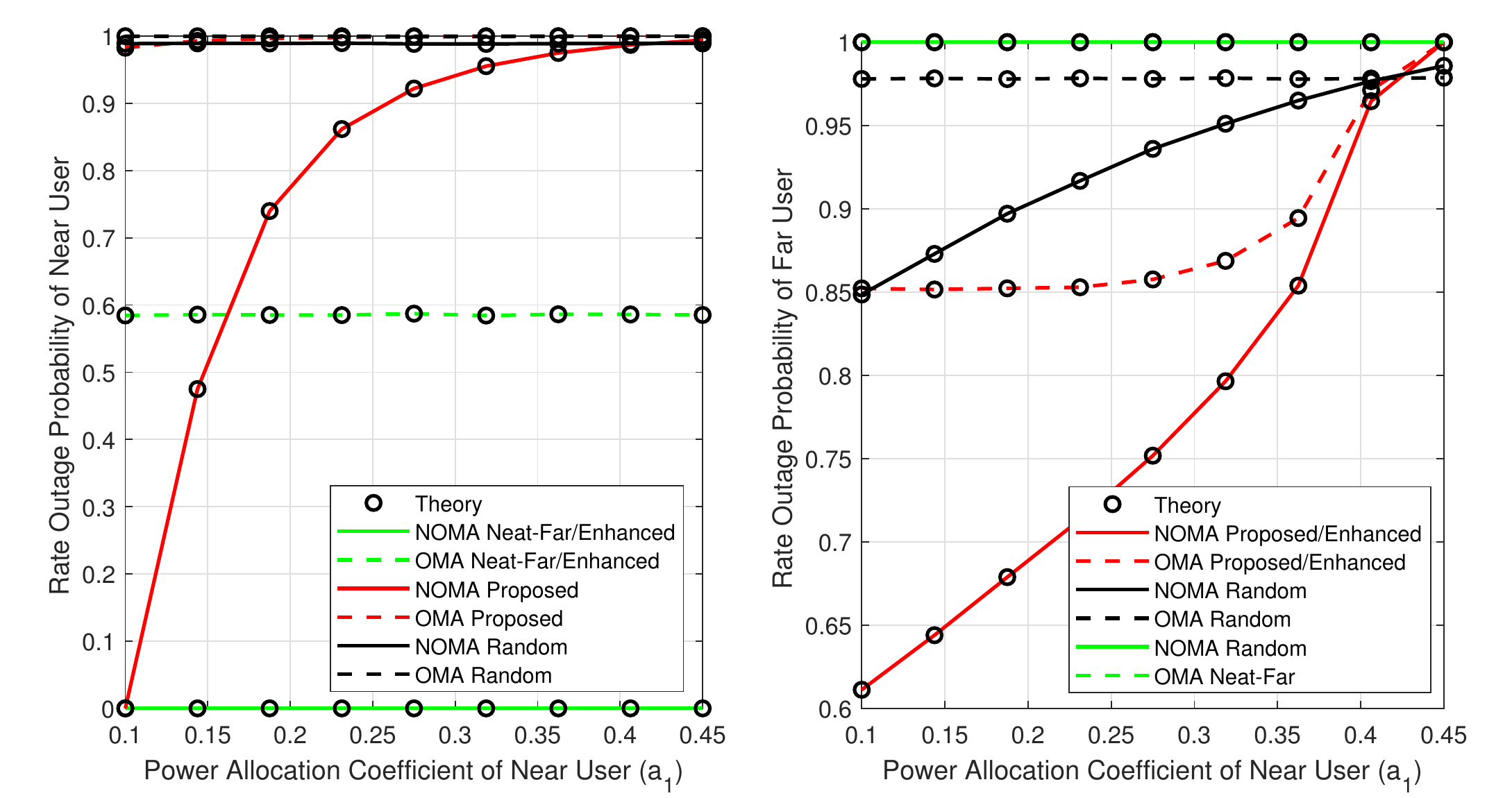}
\caption{Outage performance of  near and far users as a function of the power allocation coefficient of near user considering a single-carrier network, $k(f)=0.03$ $R$ = 60 m, $\tau_1$ = 3 bps/Hz and $\tau_2$ = 0.5 bps/Hz.}
\label{Fig2}
\end{center}
\vspace{-5 mm}
\end{figure}
Fig.~3 demonstrates the  OP of $U_1$ and $U_2$ as a function of the power allocation coefficient of near user $(a_1)$ and highlights the gain of NOMA over OMA. 
\textcolor{black}{With the increase in $a_1$, the  OP at $U_2$ increases due to the increased interference from $U_1$ and increasing values of  \( R_\mathrm{th2}\) as can be seen in \eqref{eq:14}. On the other hand, with the increase in $a_1$, the  OP at $U_1$ increases due to the increasing values of  \( R_\mathrm{th1}\) as can be seen in \eqref{eq:13}. Our scheme adapts the user pairing according to the change in $a_1$. The gain of NOMA over OMA improves for low values of $a_1$ and is significant for the enhanced and proposed schemes. Besides, we can observe that NOMA's performance gets close to OMA's when $a_1$ is close to 0.5.  }

\begin{figure}
\begin{center}
\includegraphics[totalheight=5cm,width=\linewidth]{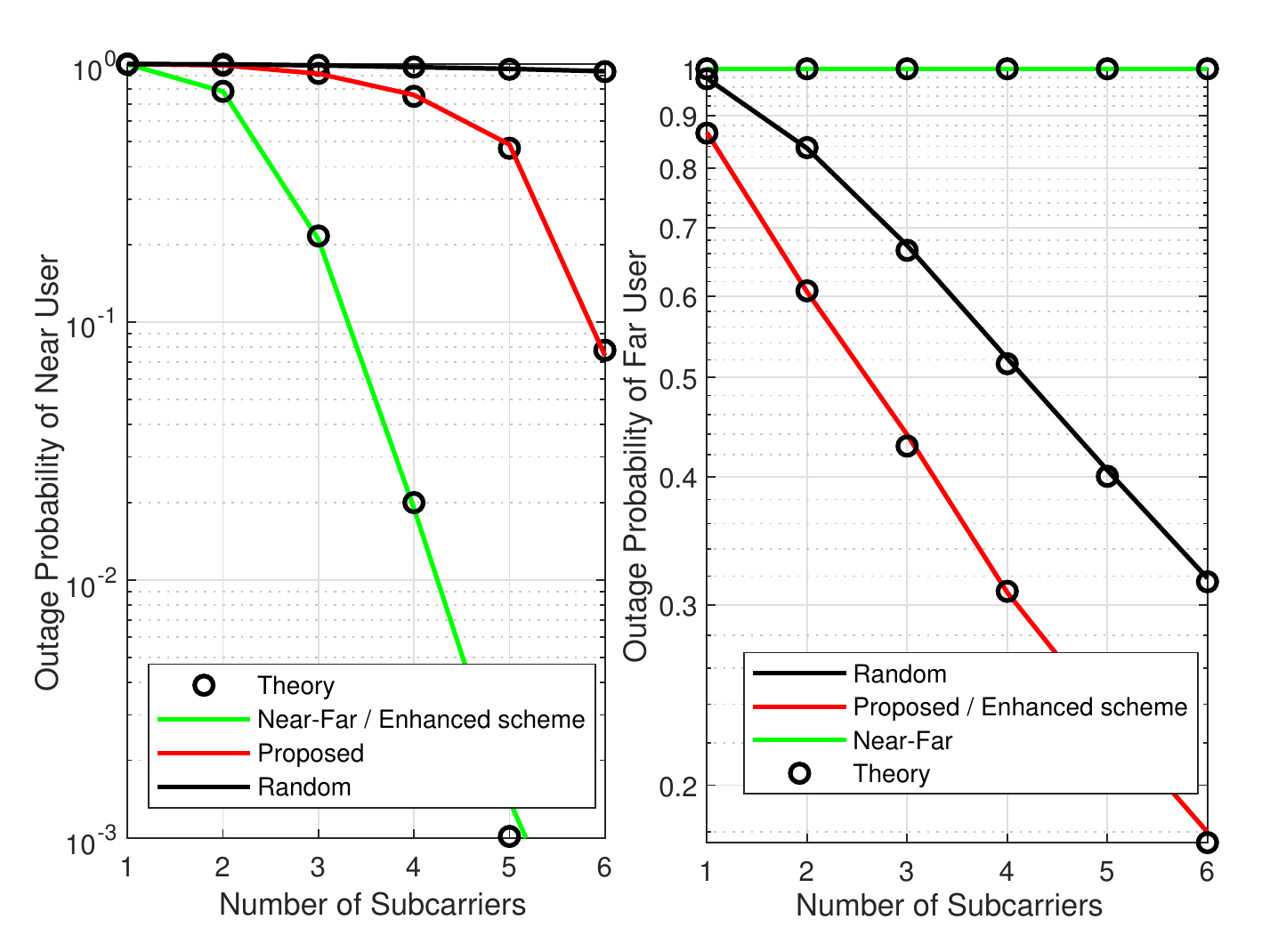}
\caption{Outage performance of near and far users as a function of the number of subcarriers in THz spectrum,  $\tau_1$=8 bps/Hz and $\tau_2$=0.5 bps/Hz.}\label{Fig3}
\end{center}
\vspace{-5 mm}
\end{figure}

\textcolor{black} {Fig.~4 shows the OP of $U_1$ and $U_2$ as a function of the  number of channels for different schemes in a multi-carrier network. The OP decreases significantly when the number of channels allocated to the user increases due to aggregate spectral efficiency. We  note that the increase in subchannels benefits the near user more as the outage decreases exponentially; whereas, the outage decreases linearly for the far user due to its channel conditions.
Our analytical results match well with the simulations, and it is evident that the enhanced scheme outperforms benchmark schemes for both users. }

\section{Conclusion}
We analyzed the performance of users  considering single-carrier and multi-carrier set-up in both THz-NOMA and THz-OMA network. The derived OP expressions are general to capture the entire range of THz spectrum, Nakagami-$m$ fading and molecular absorption noise. We have  developed an adaptive pairing scheme for THz-NOMA network where user selection adapts according to molecular absorption and  the gains of NOMA are guaranteed for each individual user.  \textcolor{black}{The framework can be extended to analyze the diversity  order or optimize network parameters, e.g., transmit power allocations.}

\bibliographystyle{IEEEtran}
\bibliography{References}
\end{document}